\documentclass[1p,final]{elsarticle}
\usepackage{amsfonts,color,morefloats,pslatex}
\usepackage{amssymb,amsthm, amsmath,latexsym}

\newtheorem{theorem}{Theorem}
\newtheorem{lemma}[theorem]{Lemma}

\newtheorem{corollary}[theorem]{Corollary}

\newtheorem{open}{Open Problem}

\newtheorem{example}{Example}
\newtheorem{conj}[theorem]{Conjecture}

\newcommand{\lcm}{{\mathrm{lcm}}}
\newcommand{\tr}{{\mathrm{Tr}}}

\newcommand{\gf}{{\mathrm{GF}}}

\newcommand{\RM}{{\mathrm{RM}}}

\newcommand{\m}{\mathbb{M}}

\newcommand{\cP}{{\mathcal{P}}}
\newcommand{\cB}{{\mathcal{B}}}

\newcommand{\C}{{\mathcal{C}}}

\newcommand{\bc}{{\mathbf{c}}}

\newcommand{\bD}{{\mathbb{D}}}

\begin{document}

\begin{frontmatter}



\title{Infinite families of $2$-designs and $3$-designs from linear codes 
\tnotetext[fn1]{C. Ding's research was supported by the Hong Kong Research Grants Council,
Proj. No. 16300415.}
}

\author[cding]{Cunsheng Ding}
\ead{cding@ust.hk}
\author[lcj]{Chengju Li}
\ead{lichengju1987@163.com}

%\cortext[lcj]{Corresponding author}
\address[cding]{Department of Computer Science
                                                  and Engineering, The Hong Kong University of Science and Technology,
                                                  Clear Water Bay, Kowloon, Hong Kong, China}
\address[lcj]{School of Computer Science and Software Engineering, East China Normal University, Shanghai, 200062, China}

\begin{abstract}
The interplay between coding theory and $t$-designs started many years ago. While every
$t$-design yields a linear code over every finite field, the largest $t$ for which an infinite
family of $t$-designs is derived directly from a linear or nonlinear code is $t=3$. Sporadic 
$4$-designs and 
$5$-designs were derived from some linear codes of certain parameters. The major objective 
of this paper is to construct many infinite families of $2$-designs and $3$-designs from linear  
codes. The parameters of some known $t$-designs are also derived. In addition, many conjectured 
infinite families of $2$-designs are also presented.  
\end{abstract}


\end{frontmatter}

\section{Introduction}

Let $\cP$ be a set of $v \ge 1$ elements, and let $\cB$ be a set of $k$-subsets of $\cP$, where $k$ is
a positive integer with $1 \leq k \leq v$. Let $t$ be a positive integer with $t \leq k$. The pair
$\bD = (\cP, \cB)$ is called a $t$-$(v, k, \lambda)$ {\em design\index{design}}, or simply {\em $t$-design\index{$t$-design}}, if every $t$-subset of $\cP$ is contained in exactly $\lambda$ elements of
$\cB$. The elements of $\cP$ are called points, and those of $\cB$ are referred to as blocks.
We usually use $b$ to denote the number of blocks in $\cB$.  A $t$-design is called {\em simple\index{simple}} if $\cB$ does not contain repeated blocks. In this paper, we consider only simple 
$t$-designs.  A $t$-design is called {\em symmetric\index{symmetric design}} if $v = b$. It is clear that $t$-designs with $k = t$ or $k = v$ always exist. Such $t$-designs are {\em trivial}. In this paper, we consider only $t$-designs with $v > k > t$.
A $t$-$(v,k,\lambda)$ design is referred to as a {\em Steiner system\index{Steiner system}} if $t \geq 2$ and $\lambda=1$, and is denoted by $S(t,k, v)$.

A necessary condition for the existence of a $t$-$(v, k, \lambda)$ design is that
\begin{eqnarray}\label{eqn-tdesignnecessty}
\binom{k-i}{t-i} \mbox{ divides } \lambda \binom{v-i}{t-i}
\end{eqnarray}
for all integer $i$ with $0 \leq i \leq t$.

There has been an interplay between codes and $t$-designs for decades. The incidence 
matrix of any $t$-design spans a linear code over any finite field $\gf(q)$. A lot of 
progress in this direction has been made and documented in the literature (see, for 
examples, \cite{AK92}, \cite{DingBook}, \cite{Tonchev,Tonchevhb}). On the other hand, 
both linear and nonlinear codes may hold $t$-designs. Some linear and nonlinear codes 
were employed to construct $2$-designs and $3$-designs \cite{AK92,Tonchev,Tonchevhb}. 
Binary and ternary Golay codes of certain parameters hold $4$-designs and $5$-designs \cite{AK92}. However, the largest 
$t$ for which an infinite family of $t$-designs is derived directly from codes is $t=3$. 
It looks that not much progress on the construction of $t$-designs from codes has 
been made so far, while many other constructions of $t$-designs are documented in the 
literature (\cite{BJL,CMhb,KLhb,RR10}).    

The main objective of this paper is to construct infinite families of $2$-designs and 
$3$-designs from linear codes. In addition, we determine the parameters of some known 
$t$-designs, and present many conjectured infinite families of $2$-designs that are based 
on projective ternary cyclic codes. 

\section{The classical construction of $t$-designs from codes and highly nonlinear functions}

Let $\C$ be a $[v, \kappa, d]$ linear code over $\gf(q)$. Let $A_i:=A_i(\C)$, which denotes the
number of codewords with Hamming weight $i$ in $\C$, where $0 \leq i \leq v$. The sequence 
$(A_0, A_1, \cdots, A_{v})$ is
called the \textit{weight distribution} of $\C$, and $\sum_{i=0}^v A_iz^i$ is referred to as
the \textit{weight enumerator} of $\C$. For each $k$ with $A_k \neq 0$, let $\cB_k$ denote
the set of supports of all codewords of Hamming weight $k$ in $\C$, where the coordinates of a codeword
are indexed by $(0,1,2, \cdots, v-1)$. Let $\cP=\{0, 1, 2, \cdots, v-1\}$.  The pair $(\cP, \cB_k)$
may be a $t$-$(v, k, \lambda)$ design for some positive integer $\lambda$. The following
theorems, developed by Assumus and Mattson, show that the pair $(\cP, \cB_k)$ defined by 
a linear code is a $t$-design under certain conditions.

\begin{theorem}\label{thm-AM1}[Assmus-Mattson Theorem \cite{AM74}, \cite[p. 303]{HP03}] 
Let $\C$ be a binary $[v, \kappa, d]$ code. Suppose $\C^\perp$ has minimum weight $d^\perp$.
Suppose that $A_i=A_i(\C)$ and $A_i^\perp=A_i(\C^\perp)$, for $0 \leq i \leq v$, are the
weight distributions of $\C$ and $\C^\perp$, respectively. Fix a positive integer $t$
with $t < d$, and let $s$ be the number of $i$ with $A_i^\perp \ne 0$ for $0 < i \leq v-t$.
Suppose that $s \leq d -t$. Then
\begin{itemize}
\item the codewords of weight $i$ in $\C$ hold a $t$-design provided that $A_i \ne 0$ and
      $d \leq i \leq v$, and
\item the codewords of weight $i$ in $\C^\perp$ hold a $t$-design provided that
      $A_i^\perp \ne 0$ and $d^\perp \leq i \leq v$.
\end{itemize}
\end{theorem}

The Assmus-Mattson Theorem for nonbinary codes is given as follows [Assmus-Mattson Theorem \cite{AM74}, \cite[p. 303]{HP03}]

\begin{theorem}\label{thm-AM2} 
Let $\C$ be a $[v, \kappa, d]$ code over $\gf(q)$. Suppose $\C^\perp$ has minimum weight $d^\perp$.
Let $w$ be the largest integer with $w \leq v$ satisfying
$$
w - \left\lfloor \frac{w+q-2}{q-1} \right\rfloor < d.
$$
(So $w=v$ when $q=2$.) Define $w^\perp$ analogously using $d^\perp$.

Suppose that $A_i=A_i(\C)$ and $A_i^\perp=A_i(\C^\perp)$, for $0 \leq i \leq v$, are the
weight distributions of $\C$ and $\C^\perp$, respectively. Fix a positive integer $t$
with $t < d$, and let $s$ be the number of $i$ with $A_i^\perp \ne 0$ for $0 < i \leq v-t$.
Suppose that $s \leq d -t$. Then
\begin{itemize}
\item the codewords of weight $i$ in $\C$ hold a $t$-design provided that $A_i \ne 0$ and
      $d \leq i \leq w$, and
\item the codewords of weight $i$ in $\C^\perp$ hold a $t$-design provided that
      $A_i^\perp \ne 0$ and $d^\perp \leq i \leq \min\{v-t, w^\perp\}$.
\end{itemize}
\end{theorem}

The Assmus-Mattson Theorems documented above are very powerful tools in constructing $t$-designs 
from linear codes. We will employ them heavily in this paper. It should be noted that the conditions 
in Theorems \ref{thm-AM1} and \ref{thm-AM2} are sufficient, but not necessary for obtaining 
$t$-designs. 

To construct $t$-designs via Theorems \ref{thm-AM1} and \ref{thm-AM2}, we will need the 
following lemma in subsequent sections, which is a variant of the MacWilliam Identity 
\cite[p. 41]{vanLint}. 

\begin{theorem} \label{thm-MI}
Let $\C$ be a $[v, \kappa, d]$ code over $\gf(q)$ with weight enumerator $A(z)=\sum_{i=0}^v A_iz^i$ and let
$A^\perp(z)$ be the weight enumerator of $\C^\perp$. Then
$$A^\perp(z)=q^{-\kappa}\Big(1+(q-1)z\Big)^vA\Big(\frac {1-z} {1+(q-1)z}\Big).$$
\end{theorem}

A function $f$ from $\gf(q^m)$ to itself is called  {\em planar} or \textit{perfect nonlinear} (PN) if
  \[\max_{0\ne a\in\gf(q^m)}\max_{b\in\gf(q^m)}|\{x\in\gf(q^m): f(x+a)-f(x)=b\}|=1,\]
and {\em almost perfect nonlinear} (APN) if
  \[\max_{0\ne a\in\gf(q^m)}\max_{b\in\gf(q^m)}|\{x\in\gf(q^m): f(x+a)-f(x)=b\}|=2.\]
Later in this paper, we will employ such functions in the constructions of linear codes and thus 
our constructions of $t$-designs.

\section{Infinite families of $3$-designs from the binary RM codes}\label{sec-brmdesigns}

It was known that Reed-Muller codes give families of $3$-$(2^m, k, \lambda)$ designs 
(\cite[Chapter 15]{MS77}, \cite{Tonchevhb}). However, the parameters of $k$ and $\lambda$ 
may not be specifically given in the literature. The purpose of this section is to determine 
the parameters of some $3$-designs derived from binary Reed-Muller codes.

We use $\RM(r, m)$ to denote the binary Reed-Muller code of length $2^m$ and order $r$.
Note that $\RM(m-r, m)^\perp=\RM(r-1, m)$, where $2 \leq r < m$. The definition and 
information about binary Reed-Muller codes can be found in \cite[Section 4.5]{vanLint}  
and \cite[Chapters 13 and 14]{MS77}.

\begin{lemma}\label{lem-RMwt1}
The weight distribution of $\RM(m-2, m)$ (except $A_i=0$) is given by
\begin{eqnarray*}
A_{4k}= \frac{1}{2^{m+1}} \left[2\binom{2^m}{4k} + (2^{m+1}-2) \binom{2^{m-1}}{2k}\right]
\end{eqnarray*}
for $0 \leq k \leq 2^{m-2}$, and
by
\begin{eqnarray*}
A_{4k+2}= \frac{1}{2^{m+1}} \left[2\binom{2^m}{4k+2} - (2^{m+1}-2) \binom{2^{m-1}}{2k+1}\right]
\end{eqnarray*}
for $0 \leq k \leq 2^{m-2}-1$.
\end{lemma}

\begin{proof}
It is well known that the weight enumerator of $\RM(1, m)$ is
$$
1 + (2^{m+1}-2)z^{2^{m-1}} + z^{2^m}.
$$
By Theorem \ref{thm-MI}, the weight enumerator of RM$(m-2, m)$, which is the dual
of RM$(1, m)$, is given by
\begin{eqnarray*}
B(z) &=& \frac{1}{2^{m+1}} (1+z)^{2^m}\left[ 1 +
             (2^{m+1}-2) \left(\frac {1-z} {1+z}\right)^{2^{m-1}} + 
                 \left(\frac {1-z} {1+z}\right)^{2^m} \right] \\
&=& \frac{1}{2^{m+1}} \left[ (1+z)^{2^m} +
             (2^{m+1}-2) (1-z^2)^{2^{m-1}} + (1-z)^{2^m} \right] \\
&=& \frac{1}{2^{m+1}} \left[ 2 \sum_{i=0}^{2^{m-1}} \binom{2^m}{2i} z^{2i}
        +(2^{m+1}-2) \sum_{i=0}^{2^{m-1}} \binom{2^{m-1}}{i} (-1)^i z^{2i}\right] \\
&=&\frac{1}{2^{m+1}} \sum_{k=0}^{2^{m-2}} \left[ 2\binom{2^m}{4k} + (2^{m+1}-2) \binom{2^{m-1}}{2k} \right]
     z^{4k} + \\
& & \frac{1}{2^{m+1}}\sum_{k=0}^{2^{m-2}-1} \left[ 2\binom{2^m}{4k+2} - (2^{m+1}-2) \binom{2^{m-1}}{2k+1}\right]
     z^{4k+2}.
\end{eqnarray*}
The desired conclusion then follows.
\end{proof}

The following theorem gives parameters of all the $3$-designs in both $\RM(m-2, m)$ and $\RM(1, m)$.

\begin{theorem}\label{thm-brmdesign1}
Let $m \geq 3$. Then $\RM(m-2, m)$ has dimension $2^m -m-1$ and minimum distance $4$.
For even positive integer $\kappa$ with $4 \leq \kappa \leq 2^m-4$, the supports of
the codewords with weight $\kappa$ in $\RM(m-2, m)$ hold a $3$-$(2^m, \kappa, \lambda)$
design, where
\begin{eqnarray*}
\lambda=\left\{
\begin{array}{ll}
\frac{\frac{1}{2^{m+1}}\binom{\kappa}{3}\left(2\binom{2^m}{4k} + (2^{m+1}-2) \binom{2^{m-1}}{2k} \right)}{\binom{2^m}{3}}  & \mbox{ if } \kappa = 4k, \\
\frac{\frac{1}{2^{m+1}}\binom{\kappa}{3}\left(2\binom{2^m}{4k+2} - (2^{m+1}-2) \binom{2^{m-1}}{2k+1} \right)}{\binom{2^m}{3}}  & \mbox{ if } \kappa = 4k+2.
\end{array}
\right.
\end{eqnarray*}

The supports of all codewords of weight $2^{m-1}$ in $\RM(1, m)$ hold a $3$-$(2^m, 2^{m-1}, 2^{m-2}-1)$
design.

\end{theorem}

\begin{proof}
Note that the weight distribution of $\RM(1,m)$ is given by
$$
A_0=1, \ A_{2^m}=1, \ A_{2^{m-1}}=2^{m+1}-2, \ \mbox{ and } A_i =0 \mbox{ for all other $i$.}
$$
It is known that the minimum distance $d$ of $\RM(m-2, m)$ is equal to $4$. Put $t=3$. The number of
$i$ with $A_i^\perp \neq 0$ and $1 \leq i \leq 2^m -3$ is $s=1$. Hence, $s=d-t$. Notice that two 
binary vectors have the same support if and only if they are equal. 
The desired conclusions then follow from Theorem \ref{thm-AM1} and Lemma \ref{lem-RMwt1}.
\end{proof}

As a corollary of Theorem \ref{thm-brmdesign1}, we have the following \cite[p. 63]{MS77}, which is well 
known.

\begin{corollary}
The minimum weight codewords in $\RM(m-2, m)$ form a $3$-$(2^m, 4, 1)$ design, i.e., a Steiner system.
\end{corollary}

The following theorem is also well known, and tells us that Reed-Muller codes give much more $3$-designs 
\cite{Tonchevhb}. 

\begin{theorem}\label{thm-brmdesign2}
Let $m \geq 4$ and $2 \leq r < m$. Then $\RM(m-r, m)$ has dimension $2^m -\sum_{i=0}^{r-1} \binom{m}{i}$ and
minimum distance $2^r$. For every nonzero weight $\kappa$ in $\RM(m-r, m)$, the codewords
of weight $\kappa$ in RM$(m-r, m)$ hold a $3$-$(2^m, \kappa, \lambda)$ design.
\end{theorem}

\begin{proof}
Since $2 \leq r < m$, by Theorem 24 in \cite[p. 400]{MS77}, the automorphism group of $\RM(m-r, m)$ 
is triply transitive. The desired conclusion then follows from Theorem 8.4.7 in \cite[p. 308]{HP03}. 
\end{proof}

Determining the weight distribution of $\RM(m-r, m)$ may be hard for $3 \leq r \leq m-3$ in general. 
Therefore, it may be difficult to find out the parameters $(\kappa, \lambda)$ of all the $3$-designs. 
The following problem is open in general.  

\begin{open} 
Determine the weight distribution of $\RM(m-r, m)$ for $3 \leq r \leq m-3$. 
\end{open} 

Some progress on the open problem above was made by Kasami and Tokura \cite{KT} and 
Kasami, Tokura and Azumi \cite{KTA}. Detailed information on this problem can be found 
in \cite[Chapter 15]{MS77}.

\section{Designs from cyclic Hamming codes}\label{sec-hmdesigns}

Let $\alpha$ be a generator of $\gf(q^m)^*$. Set $\beta=\alpha^{q-1}$. Let $g(x)$ be the minimal polynomial of $\beta$ over
$\gf(q)$. Let $\C_{(q,m)}$ denote the cyclic code of length $v=(q^m-1)/(q-1)$ over $\gf(q)$ with generator polynomial $g(x)$.
Then $\C_{(q,m)}$ has parameters $[(q^m-1)/(q-1), (q^m-1)/(q-1)-m, d]$, where $d \in \{2,3\}$.
When $\gcd(q-1, m)=1$, $\C_{(q,m)}$ has minimum weight $3$ and is equivalent to the Hamming code.

\begin{lemma}\label{lem-HCwt}
The weight distribution of $\C_{(q,m)}$ is given by
\begin{eqnarray*}
A_{k}= \frac{1}{q^m} \sum_{\substack {0 \le i \le (q^{m-1}-1)/(q-1) \\0 \le j \le q^{m-1} \\ i+j=k}}\left[\binom{\frac{q^{m-1}-1}{q-1}}{i}
 \binom{q^{m-1}}{j}\Big((q-1)^k+(-1)^j(q-1)^i(q^m-1)\Big)\right]
\end{eqnarray*}
for $0 \leq k \leq (q^m-1)/(q-1)$.
\end{lemma}

\begin{proof}
$\C_{(q,m)}^\perp$ is the simplex code, as $\gcd(q-1, (q^{m}-1)/(q-1))=1$. Its weight enumerator
is
$$
1+(q^m-1)z^{q^{m-1}}.
$$
By Theorem \ref{thm-MI}, the weight enumerator of $\C_{(q,m)}$ is given by
\begin{eqnarray*}
A(z) &=& \frac{1}{q^m} (1+(q-1)z)^{v}\left[ 1 +
             (q^{m}-1) \left(\frac {1-z} {1+(q-1)z}\right)^{q^{m-1}} \right] \\
&=& \frac{1}{q^m} \left[ (1+(q-1)z)^{v} +
             (q^m-1) (1-z)^{q^{m-1}}(1+(q-1)z)^{\frac {q^{m-1}-1}{q-1}} \right] \\
&=& \frac{1}{q^m} (1+(q-1)z)^{\frac {q^{m-1}-1}{q-1}} \left[ (1+(q-1)z)^{q^{m-1}} +
             (q^m-1) (1-z)^{q^{m-1}} \right].
\end{eqnarray*}
The desired conclusion then follows.
\end{proof}

A code of minimum distance $d=2e+1$ is \textit{perfect}, if the spheres of radius 
$e$ around the codewords cover the whole space. The following theorem introduces 
a relation between perfect codes and $t$-designs and is due to Assmus and Mattson 
\cite{AM74}. 

\begin{theorem}\label{thm-perfectcodedesign}
A linear $q$-ary code of length $v$ and minimum distance $d=2e+1$ is perfect if and 
only if the supports of the codewords of minimum weight form a simple $(e+1)$-$(v, 2e+1, 
(q-1)^e)$ design. In particular, the minimum weight codewords in a linear or nonlinear 
perfect code, which contains the zero vector, form a Steiner system $S(e+1, 2e+1, v)$. 
\end{theorem}

It is known that the Hamming code over $\gf(q)$ is perfect, and the codewords of 
weight $3$ hold a $2$-design by Theorem \ref{thm-perfectcodedesign}. The $2$-designs 
documented in the following theorem may be viewed as an extension of this result. 

\begin{theorem}\label{thm-HMdesign171}
Let $m \geq 3$ and $q = 2$ or $m \geq 2$ and $q >2$,  and let $\gcd(q-1, m)=1$. 
Let $\cP=\{0,1,2, \cdots, (q^m-q)/(q-1)\}$, and let $\cB$ be the set of the supports of the codewords of
 Hamming weight $k$ with $A_k \neq 0$ in $\C_{(q,m)}$, where $3 \leq k \leq w$ and $w$ is the 
 largest such that $w-\lfloor (w+q-2)/(q-1) \rfloor < 3$. Then $(\cP, \cB)$ 
is a $2$-$((q^m-1)/(q-1), k, \lambda)$ design. In particular, the supports 
of codewords of weight $3$ in $\C_{(q,m)}$ form a $2$-$((q^m-1)/(q-1), 3, q-1)$ 
design.  

The supports of all codewords of weight $q^{m-1}$ in $\C_{(q,m)}^\perp$ form a 
$2$-$((q^m-1)/(q-1), q^{m-1}, \lambda)$ design, where 
$$ 
\lambda=(q-1)q^{m-2}. 
$$

\end{theorem}

\begin{proof}
$\C_{(q,m)}^\perp$ is the simplex code, as $\gcd(q-1, (q^{m}-1)/(q-1))=1$. Its weight enumerator
is
$$
1+(q^m-1)z^{q^{m-1}}.
$$
A proof of this weight enumerator is straightforward and can be found in \cite[Theorem 15]{DY13}. 

Recall now Theorem \ref{thm-AM2} and the definition of $w$ for $\C_{(q,m)}$ 
and $w^\perp$ for $\C_{(q,m)}^\perp$. Since $\C_{(q,m)}$ has minimum weight 
$3$. Given that the weight enumerator of $\C_{(q,m)}^\perp$ is 
$1+(q^m-1)z^{q^{m-1}},$ we deduce that $w^\perp=q^{m-1}$. Put $t=2$. It then 
follows that $s=1=d-t$. The desired conclusion on the $2$-design property 
then follows from Theorem \ref{thm-AM2} and Lemma \ref{lem-HCwt}. 

We now prove that the supports of codewords of weight $3$ in $\C_{(q,m)}$ 
form a $2$-$((q^m-1)/(q-1), 3, q-1)$ design. We have already proved that 
these supports form a $2$-$((q^m-1)/(q-1), 3, \lambda)$ design. To determine the 
value $\lambda$ for this design, we need to compute the total number $b$ 
of blocks in this design. To this end, we first compute the total number 
of codewords of weight $3$ in $\C_{(q,m)}$. It follows from Lemma \ref{lem-HCwt} 
that 
%\begin{eqnarray*}
%q^m A_3 
%&=& \binom{\frac{q^{m-1}-1}{q-1}}{0} \binom{q^{m-1}}{3} [(q-1)^3 - (q-1)^0(q^m-1)] + \\
%& & \binom{\frac{q^{m-1}-1}{q-1}}{1} \binom{q^{m-1}}{1} [(q-1)^3 - (q-1)^1(q^m-1)] + \\
%& & \binom{\frac{q^{m-1}-1}{q-1}}{2} \binom{q^{m-1}}{1} [(q-1)^3 - (q-1)^2(q^m-1)] + \\ 
%& & \binom{\frac{q^{m-1}-1}{q-1}}{3} \binom{q^{m-1}}{0} [(q-1)^3 - (q-1)^3(q^m-1)]  \\ 
%&=& q^m(q^m-1)(q^m-q).   
%\end{eqnarray*}  
%We obtain then 
$$ 
A_3=\frac{(q^m-1)(q^m-q)}{6}. 
$$  
Since $3$ is the minimum nonzero weight in $\C_{(q,m)}$, it is easy to see that 
two codewords of weight $3$ in $\C_{(q,m)}$ have the same support if and only 
one is a scalar multiple of another. Thus, the total number $b$ of blocks is given by  
$$ 
b:=\frac{A_3}{q-1}=\frac{(q^m-1)(q^m-q)}{6(q-1)}. 
$$ 
It then follows that 
$$ 
\lambda=\frac{b\binom{3}{2}}{\binom{\frac{q^m-1}{q-1}}{2}}=q-1. 
$$

Let $\alpha$ be a generator of $\gf(q^m)^*$, and set $\beta=\alpha^{q-1}$. 
Then $\beta$ is a $v$-th primitive root of unity, where $v=(q^m-1)/(q-1)$. 
It is known that 
$$ 
\C_{(q,m)}^\perp=\{\bc_u: u \in \gf(q^m)\},   
$$ 
where $\bc_u=((\tr(u), \tr(u\beta), \cdots, \tr(u\beta^{v-1}))$ and $\tr(x)$ 
is the trace function from $\gf(q^m)$ to $\gf(q)$. It is then easily seen that 
$\bc_{u}$ and $\bc_{v}$ have the same support if and only if $u=av$ for some 
$a \in \gf(q)^*$. We then deduce that the total number $b^\perp$ of blocks in the design 
is given by 
$$ 
b^\perp = \frac{q^m-1}{q-1}. 
$$ 
Consequently, 
$$ 
\lambda^\perp = \frac{\frac{q^m-1}{q-1} \binom{q^{m-1}}{2}}{\binom{\frac{q^{m-1}-1}{q-1}}{2}} 
=(q-1)q^{m-2}. 
$$ 
Thus, the supports of all codewords of weight $q^{m-1}$ in 
$\C_{(q,m)}^\perp$ form a $2$-design with parameters 
$$
\left((q^m-1)/(q-1), \ q^{m-1}, \ (q-1)q^{m-2} \right). 
$$
\end{proof}

Theorem \ref{thm-HMdesign171} tells us that for some $k \geq 3$ with $A_k \neq 0$, the 
supports of the codewords with weight $k$ in $\C_{(q,m)}$ form $2$-$((q^m-1)/(q-1), k, 
\lambda)$ design. However, it looks complicated to determine the parameter $\lambda$ 
corresponding to this $k \geq 4$. We draw the reader's attention to the following open problem.  

\begin{open} 
Let $q \geq 3$ and $m \geq 2$. For $k \geq 4$ with $A_k \neq 0$, determine the value $\lambda$ in the 
$2$-$((q^m-1)/(q-1), k, \lambda)$ design, formed by the supports of the codewords with 
weight $k$ in $\C_{(q,m)}$. 
\end{open} 

Notice that two binary codewords have the same support if and only if they are equal. 
When $q=2$, Theorem \ref{thm-HMdesign171} becomes the following. 

\begin{corollary}\label{cor-HMdesign171} 
Let $m \geq 3$. Let $\cP=\{0,1,2, \cdots, 2^m-2\}$, and let $\cB$ be  the set of the supports 
of the codewords with Hamming weight $k$ in $\C_{(2,m)}$, where $3 \leq k \leq 2^m-3$. 
Then $(\cP, \cB)$ is a $2$-$(2^m-1, k, \lambda)$ design, where 
$$ 
\lambda=\frac{(k-1)kA_k}{(2^m-1)(2^m-2)}  
$$ 
and $A_k$ is given in Lemma \ref{lem-HCwt}.  

The supports of all codewords of weight $2^{m-1}$ in $\C_{(2,m)}^\perp$ form a 
$2$-$(2^m-1, 2^{m-1}, 2^{m-2})$ design. 
\end{corollary} 

Corollary \ref{cor-HMdesign171} says that each binary Hamming code $\C_{(2,m)}$ 
and its dual code give a total number $2^m-4$ of $2$-designs with varying block 
sizes. 

The following are examples of the $2$-designs held in the binary Hamming code. 

\begin{example} 
Let $m \geq 4$. Let $\cP=\{0,1,2, \cdots, 2^m-2\}$, and let $\cB$ be  the set of the supports 
of the codewords with Hamming weight $3$ in $\C_{(2,m)}$.  
Then $(\cP, \cB)$ is a $2$-$(2^m-1, \, 3, \, 1)$ design. 
\end{example} 

\begin{proof}
By Lemma \ref{lem-HCwt}, we have 
$$ 
A_3=\frac{(2^{m-1}-1)(2^m-1)}{3}. 
$$
The desired value for $\lambda$ then follows from Corollary \ref{cor-HMdesign171}. 
\end{proof}

\begin{example}\label{exam-hmdesign4} 
Let $m \geq 4$. Let $\cP=\{0,1,2, \cdots, 2^m-2\}$, and let $\cB$ be  the set of the supports 
of the codewords with Hamming weight $4$ in $\C_{(2,m)}$.  
Then $(\cP, \cB)$ is a $2$-$(2^m-1, \, 4, \, 2^{m-1}-2)$ design. 
\end{example} 

\begin{proof}
By Lemma \ref{lem-HCwt}, we have 
$$ 
A_4=\frac{(2^{m-1}-1)(2^{m-1}-2)(2^m-1)}{6}. 
$$
The desired value for $\lambda$ then follows from Corollary \ref{cor-HMdesign171}. 
\end{proof}

\begin{example}\label{exam-hmdesign5}  
Let $m \geq 4$. Let $\cP=\{0,1,2, \cdots, 2^m-2\}$, and let $\cB$ be  the set of the supports 
of the codewords with Hamming weight $5$ in $\C_{(2,m)}$.  
Then $(\cP, \cB)$ is a $2$-$(2^m-1, \, 5, \, \lambda)$ design, where 
$$ 
\lambda=\frac{2(2^{m-1}-2)(2^{m-1}-4)}{3} 
$$  
\end{example} 

\begin{proof}
By Lemma \ref{lem-HCwt}, we have 
$$ 
A_5=\frac{(2^{m-1}-1)(2^{m-1}-2)(2^{m-1}-4)(2^m-1)}{15}. 
$$
The desired value for $\lambda$ then follows from Corollary \ref{cor-HMdesign171}. 
\end{proof}

\begin{example}\label{exam-hmdesign6}  
Let $m \geq 4$. Let $\cP=\{0,1,2, \cdots, 2^m-2\}$, and let $\cB$ be  the set of the supports 
of the codewords with Hamming weight $6$ in $\C_{(2,m)}$.  
Then $(\cP, \cB)$ is a $2$-$(2^m-1, \, 6, \, \lambda)$ design, where 
$$ 
\lambda=\frac{(2^{m-1}-2)(2^{m-1}-3)(2^{m-1}-4)}{3} 
$$  
\end{example} 

\begin{proof}
By Lemma \ref{lem-HCwt}, we have 
$$ 
A_6=\frac{(2^{m-1}-1)(2^{m-1}-2)(2^{m-1}-3)(2^{m-1}-4)(2^m-1)}{45}. 
$$
The desired value for $\lambda$ then follows from Corollary \ref{cor-HMdesign171}. 
\end{proof}

\begin{example}\label{exam-hmdesign7}  
Let $m \geq 4$. Let $\cP=\{0,1,2, \cdots, 2^m-2\}$, and let $\cB$ be  the set of the supports 
of the codewords with Hamming weight $7$ in $\C_{(2,m)}$.  
Then $(\cP, \cB)$ is a $2$-$(2^m-1, \, 7, \, \lambda)$ design, where 
$$ 
\lambda=\frac{(2^{m-1}-2)(2^{m-1}-3)(4 \times 2^{2(m-1)}-30 \times 2^{m-1} +71)}{30}.   
$$  
\end{example} 

\begin{proof}
By Lemma \ref{lem-HCwt}, we have 
$$ 
A_7=\frac{(2^{m-1}-1)(2^{m-1}-2)(2^{m-1}-3)(2^m-1)(4 \times 2^{2(m-1)}-30 \times 2^{m-1} +71)}{630}. 
$$
The desired value for $\lambda$ then follows from Corollary \ref{cor-HMdesign171}. 
\end{proof}

\section{Designs from a class of binary codes with two zeros and their duals}\label{sec-newdesigns}

In this section, we construct many infinite families of $2$-designs and $3$-designs with 
several classes of binary cyclic codes whose duals have two zeros. These binary codes are 
defined by almost perfect nonlinear (APN) functions over $\gf(2^m)$. 

\begin{table}[ht]
\caption{Weight distribution for odd $m$.}\label{tab-CG1}
\centering
\begin{tabular}{ll}
\hline
Weight $w$    & No. of codewords $A_w$  \\ \hline
$0$                                                        & $1$ \\
$2^{m-1}-2^{(m-1)/2}$           & $(2^m-1)(2^{(m-1)/2}+1)2^{(m-3)/2}$ \\
$2^{m-1}$                             & $(2^m-1)(2^{m-1}+1)$ \\
$2^{m-1}+2^{(m-1)/2}$           & $(2^m-1)(2^{(m-1)/2}-1)2^{(m-3)/2}$ \\ \hline
\end{tabular}
\end{table}

\begin{lemma}\label{lem-TZwt}
Let $m \geq 5$ be odd. 
Let $\C_m$ be a binary linear code of length $2^m-1$ such that its dual code
$\C_m^\perp$ has the weight distribution of Table \ref{tab-CG1}. Then the weight distribution
of $\C_m$ is given by
\begin{eqnarray*}
2^{2m}A_k&=& \sum_{\substack{0 \le i \le 2^{m-1}-2^{(m-1)/2} \\
0\le j \le 2^{m-1}+2^{(m-1)/2}-1 \\i+j=k}}(-1)^ia\binom{2^{m-1}-2^{(m-1)/2}} {i} \binom{2^{m-1}+2^{(m-1)/2}-1}{j}\\
& & + \binom {2^m-1}{k}+\sum_{\substack{0 \le i \le 2^{m-1} \\
0\le j \le 2^{m-1}-1 \\i+j=k}}(-1)^ib\binom{2^{m-1}} {i}\binom{2^{m-1}-1} {j} \\
& & + \sum_{\substack{0 \le i \le 2^{m-1}+2^{(m-1)/2} \\
0\le j \le 2^{m-1}-2^{(m-1)/2}-1 \\i+j=k}}(-1)^ic\binom{2^{m-1}+2^{(m-1)/2}}{i}\binom{2^{m-1}-2^{(m-1)/2}-1}{j} 
\end{eqnarray*}
for $0 \le k \le 2^m-1$, where 
\begin{eqnarray*}
a &=& (2^m-1)(2^{(m-1)/2}+1)2^{(m-3)/2}, \\   
b &=& (2^m-1)(2^{m-1}+1), \\  
c &=& (2^m-1)(2^{(m-1)/2}-1)2^{(m-3)/2}.
\end{eqnarray*} 
In addition, $\C_m$ has parameters $[2^m-1, 2^m-1-2m, 5]$. 
\end{lemma}

\begin{proof}
By assumption, the weight enumerator of $\C_m^\perp$ 
is given by 
$$
A^\perp(z)=1+az^{2^{m-1}-2^{(m-1)/2}}+bz^{2^{m-1}}+cz^{2^{m-1}+2^{(m-1)/2}}.
$$
It then follows from Theorem \ref{thm-MI} that the weight enumerator of $\C_m$ is given by
\begin{eqnarray*}
A(z) &=& \frac{1}{2^{2m}} (1+z)^{2^m-1}\left[ 1 +
             a\left(\frac {1-z} {1+z}\right)^{2^{m-1}-2^{(m-1)/2}} \right] + \\
& & \frac{1}{2^{2m}} (1+z)^{2^m-1}\left[
             b\left(\frac {1-z} {1+z}\right)^{2^{m-1}}+c\left(\frac {1-z} {1+z}\right)^{2^{m-1}+2^{(m-1)/2}} \right] \\             
&=& \frac{1}{2^{2m}} \Bigg[ (1+z)^{2^m-1} +
             a(1-z)^{2^{m-1}-2^{(m-1)/2}}(1+z)^{2^{m-1}+2^{(m-1)/2}-1} \\ 
             & & + b(1-z)^{2^{m-1}}(1+z)^{2^{m-1}-1} + 
             c(1-z)^{2^{m-1}+2^{(m-1)/2}}(1+z)^{2^{m-1}-2^{(m-1)/2}-1} \Bigg].
\end{eqnarray*}
Obviously, we have   
\begin{eqnarray*}
(1+z)^{2^m-1} &=& \sum_{k=0}^{2^m-1} \binom{2^m-1}{k}z^k. 
\end{eqnarray*}  
It is easily seen that 
\begin{eqnarray*}
\lefteqn{(1-z)^{2^{m-1}-2^{(m-1)/2}}(1+z)^{2^{m-1}+2^{(m-1)/2}-1} } \\  
&=&  \sum_{k=0}^{2^m-1} \left[ \sum_{\substack{0 \le i \le 2^{m-1}-2^{(m-1)/2} \\
0\le j \le 2^{m-1}+2^{(m-1)/2}-1 \\i+j=k}}(-1)^i  \binom{2^{m-1}-2^{(m-1)/2}} {i} \binom{2^{m-1}+2^{(m-1)/2}-1}{j} \right] z^k  
\end{eqnarray*} 
and 
\begin{eqnarray*}
\lefteqn{(1-z)^{2^{m-1}+2^{(m-1)/2}}(1+z)^{2^{m-1}-2^{(m-1)/2}-1} } \\
&=&  \sum_{k=0}^{2^m-1}  \left[ \sum_{\substack{0 \le i \le 2^{m-1}+2^{(m-1)/2} \\
0\le j \le 2^{m-1}-2^{(m-1)/2}-1 \\i+j=k}}(-1)^i \binom{2^{m-1}+2^{(m-1)/2}}{i}\binom{2^{m-1}-2^{(m-1)/2}-1}{j} \right] z^k. 
\end{eqnarray*}
Similarly, we have 
\begin{eqnarray*}
(1-z)^{2^{m-1}}(1+z)^{2^{m-1}-1}=
  \sum_{k=0}^{2^m-1} \left[ \sum_{\substack{0 \le i \le 2^{m-1} \\
0\le j \le 2^{m-1}-1 \\i+j=k}}(-1)^i \binom{2^{m-1}} {i}\binom{2^{m-1}-1} {j}  \right] z^k. 
\end{eqnarray*}
Combining these formulas above yields the weight distribution formula for $A_k$. 

The weight distribution in Table \ref{tab-CG1} tells us that the dimension of $\C_m^\perp$ 
is $2m$. Therefore, the dimension of $\C_m$ is equal to $2^m-1-2m$. Finally, we 
prove that the minimum distance $d$ of $\C_m$ equals $5$. 

After tedious computations with the formula of $A_k$ given in Lemma \ref{lem-TZwt}, 
one can verify that $A_1=A_2=A_3=A_4=0$ and  
\begin{eqnarray}\label{eqn-minimumwt5}
A_5=\frac{4\times 2^{3m-5} - 22 \times 2^{2m-4} + 26 \times 2^{m-3} -2}{15}. 
\end{eqnarray} 
When $m \geq 5$, we have  
$$
4\times 2^{3m-5} = 4 \times 2^{m-1} 2^{2m-4} \geq 64 \times 2^{2m-4} > 22 \times 2^{2m-4}
$$ 
and 
$$ 
26 \times 2^{m-3} -2 >0. 
$$
Consequently, $A_5>0$ for all odd $m$. This proves that $d=5$.

\end{proof}

\begin{theorem}\label{thm-newdesigns1} 
Let $m \geq 5$ be odd. 
Let $\C_m$ be a binary linear code of length $2^m-1$ such that its dual code
$\C_m^\perp$ has the weight distribution of Table \ref{tab-CG1}. 
Let $\cP=\{0,1,2, \cdots, 2^m-2\}$, and let $\cB$ be  the set of the supports of the codewords of $\C_m$
with weight $k$, where $A_k \neq 0$. Then $(\cP, \cB)$
is a $2$-$(2^m-1, k, \lambda)$ design, where
\begin{eqnarray*}
\lambda=\frac{k(k-1)A_k}{(2^m-1)(2^m-2)}, 
\end{eqnarray*} 
where $A_k$ is given in Lemma \ref{lem-TZwt}. 

Let $\cP=\{0,1,2, \cdots, 2^m-2\}$, and let $\cB^\perp$ be  the set of  the supports of the codewords of $\C_m^\perp$
with weight $k$ and $A_k^\perp \neq 0$. Then $(\cP, \cB^\perp)$
is a $2$-$(2^m-1, k, \lambda)$ design, where
\begin{eqnarray*}
\lambda=\frac{k(k-1)A_k^\perp}{(2^m-1)(2^m-2)}, 
\end{eqnarray*}
where $A_k^\perp$ is given in Lemma \ref{lem-TZwt}. 
\end{theorem}

\begin{proof}
The weight distribution of $\C_m$ is given in Lemma \ref{lem-TZwt} and that of $\C_m^\perp$ 
is given in Table \ref{tab-CG1}. By Lemma \ref{lem-TZwt}, the minimum distance $d$ of $\C_m$ 
is equal to $5$. Put $t=2$. The number of $i$ with $A_i^\perp \neq 0$ and $1 \leq i \leq 2^m-1 -t$ 
is $s=3$. Hence, $s=d-t$. The desired conclusions then follow from Theorem \ref{thm-AM1} and 
the fact that two binary vectors have the same support if and only if they are equal.  
\end{proof}

\begin{example} 
Let $m \geq 5$ be odd. 
Then $\C_m^\perp$ gives three $2$-designs with the following parameters: 
\begin{itemize}
\item $(v,\, k, \, \lambda)=\left(2^m-1,\  2^{m-1}-2^{(m-1)/2}, \  2^{m-3} (2^{m-1} - 2^{(m-1)/2} -1)  \right).$
\item $(v, \, k, \, \lambda)=\left(2^m-1, \ 2^{m-1}+2^{(m-1)/2}, \  2^{m-3} (2^{m-1} + 2^{(m-1)/2} -1)  \right).$ 
\item $(v, \, k, \, \lambda)=\left(2^m-1, \ 2^{m-1}, \  (2^m-1)(2^{m-1}+1) \right).$  
\end{itemize} 
\end{example}

\begin{example} 
Let $m \geq 5$ be odd. 
Then the supports of all codewords of weight $5$ in $\C_m$ give a $2$-$(2^m-1,\, 5,\, (2^{m-1}-4)/3)$ design. 
\end{example} 

\begin{proof}
By Lemma  \ref{lem-TZwt}, 
$$ 
A_5 = \frac{(2^{m-1}-1) (2^{m-1}-4) (2^m-1)}{30}
$$
The desired value for $\lambda$ then follows from Theorem \ref{thm-newdesigns1}. 
\end{proof}

\begin{example} 
Let $m \geq 5$ be odd. 
Then the supports of all codewords of weight $6$ in $\C_m$ give a $2$-$(2^m-1,\, 6,\, \lambda)$ design, where 
$$ 
\lambda= \frac{(2^{m-2}-2)(2^{m-1}-3)}{3}. 
$$
\end{example}

\begin{proof}
By Lemma  \ref{lem-TZwt}, 
$$ 
A_6 = \frac{(2^{m-1}-1) (2^{m-1}-4) (2^{m-1}-3)  (2^m-1)}{90}
$$
The desired value for $\lambda$ then follows from Theorem \ref{thm-newdesigns1}. 
\end{proof}

\begin{example} 
Let $m \geq 5$ be odd. 
Then the supports of all codewords of weight $7$ in $\C_m$ give a $2$-$(2^m-1,\, 7,\, \lambda)$ design, where 
$$ 
\lambda= \frac{2 \times 2^{3(m-1)} - 25 \times 2^{2(m-1)} + 123 \times 2^{m-1} - 190}{30}. 
$$
\end{example} 

\begin{proof}
By Lemma  \ref{lem-TZwt}, 
$$ 
A_7 = \frac{(2^{m-1}-1) (2^m-1) (2 \times 2^{3(m-1)} - 25 \times 2^{2(m-1)} + 123 \times 2^{m-1} - 190)}{630}. 
$$
The desired value for $\lambda$ then follows from Theorem \ref{thm-newdesigns1}. 
\end{proof}

\begin{example} 
Let $m \geq 5$ be odd. 
Then the supports of all codewords of weight $8$ in $\C_m$ give a $2$-$(2^m-1,\, 8,\, \lambda)$ design, where 
$$ 
\lambda= \frac{ (2^{m-2}-2)  (2 \times 2^{3(m-1)} - 25 \times 2^{2(m-1)} + 123 \times 2^{m-1} - 190)}{45}. 
$$
\end{example} 

\begin{proof}
By Lemma  \ref{lem-TZwt}, 
$$ 
A_8 =  \frac{(2^{m-1}-1) (2^{m-1}-4)  (2^m-1) (2 \times 2^{3(m-1)} - 25 \times 2^{2(m-1)} + 123 \times 2^{m-1} - 190)}{8 \times 315}. 
$$
The desired value for $\lambda$ then follows from Theorem \ref{thm-newdesigns1}. 
\end{proof}

\begin{lemma}\label{lem-TZEwt} 
Let $m \geq 5$ be odd. Let $\C_m$ be a linear code of length $2^m-1$ such that its dual code
$\C_m^\perp$ has the weight distribution of Table \ref{tab-CG1}. Denote by $\overline{\C}_m$ 
the extended code of $\C_m$ and let $\overline{\C}_m^\perp$ denote the dual of $\overline{\C}_m$.
Then the weight distribution
of $\overline{\C}_m$ is given by 
\begin{eqnarray*}
2^{2m+1}\overline{A}_k&=&  (1+(-1)^k) \binom{2^m}{k} + 
\frac{1+(-1)^k}{2} (-1)^{\lfloor k/2 \rfloor} \binom{2^{m-1}}{\lfloor k/2 \rfloor} v  + \\ 
& & u \sum_{\substack{0 \le i \le 2^{m-1}-2^{(m-1)/2} \\
0\le j \le 2^{m-1}+2^{(m-1)/2} \\i+j=k}}(-1)^i \binom{2^{m-1}-2^{(m-1)/2}} {i} \binom{2^{m-1}+2^{(m-1)/2}}{j} + \\
& & u  
\sum_{\substack{0 \le i \le 2^{m-1}+2^{(m-1)/2} \\
0\le j \le 2^{m-1}-2^{(m-1)/2} \\i+j=k}}(-1)^i \binom{2^{m-1}+2^{(m-1)/2}}{i}\binom{2^{m-1}-2^{(m-1)/2}}{j} 
\end{eqnarray*}
for $0 \le k \le 2^m$, where 
$$ 
u=2^{2m-1}-2^{m-1} \mbox{ and } v = 2^{2m}+2^m-2. 
$$
In addition, $\overline{\C}_m$ has parameters $[2^m, 2^m-1-2m, 6]$. 

The code $\overline{\C}_m^\perp$ has weight enumerator 
\begin{eqnarray}\label{eqn-wtenumerator}
\overline{A}^\perp(z) = 1+uz^{2^{m-1}-2^{(m-1)/2}}+vz^{2^{m-1}}+uz^{2^{m-1}+2^{(m-1)/2}}+z^{2^m}, 
\end{eqnarray} 
and parameters $[2^m, \ 2m+1, \ 2^{m-1}-2^{(m-1)/2}]$. 

\end{lemma}

\begin{proof} 
It was proved in Lemma \ref{lem-TZwt} that $\C_m$ has parameters $[2^m-1, 2^m-1-2m, 5]$. 
By definition, the extended code $\overline{\C}_m$ has parameters $[2^m, 2^m-1-2m, 6]$. 
By Table \ref{tab-CG1},  all weights of $\C_m^\perp$ are even. Note that $\C_m^\perp$ 
has length $2^m-1$ and dimension $2m$, while $\overline{\C}_m^\perp$ has length $2^m$ and dimension 
$2m+1$. By definition, $\overline{\C}_m$ has only even weights. Therefore, the all-one vector must 
be a codeword in $\overline{\C}_m^\perp$. It can be shown that the  weights in  $\overline{\C}_m^\perp$ 
are the following: 
$$ 
0, \ w_1,\  w_2, \ w_3,\  2^m-w_1,\  2^m-w_2, \  2^m-w_3, \ 2^m, 
$$
where $w_1, w_2$ and $w_3$ are the three nonzero weights in $\C_m^\perp$. 
Consequently,  $\overline{\C}_m^\perp$ has the following four weights 
$$ 
2^{m-1}-2^{(m-1)/2}, \ 2^{m-1}, \ 2^{m-1}+2^{(m-1)/2}, \ 2^m.  
$$ 
Recall that $\overline{\C}_m$ has minimum distance $6$. Employing 
the first few Pless Moments, one can prove that the weight enumerator of $\overline{\C}_m^\perp$ 
is the one given in (\ref{eqn-wtenumerator}).  

By Theorem \ref{thm-MI}, the weight enumerator of $\overline{\C}_m$ is given by
\begin{eqnarray}\label{eqn-j18-1}
2^{2m+1}\overline{A}(z) &=&  (1+z)^{2^m}\left[ 1 +
             u\left(\frac {1-z} {1+z}\right)^{2^{m-1}-2^{(m-1)/2}}+
             v\left(\frac {1-z} {1+z}\right)^{2^{m-1}} \right] +  \nonumber \\
      &&  (1+z)^{2^m}\left[ u\left(\frac {1-z} {1+z}\right)^{2^{m-1}+2^{(m-1)/2}} + 
      \left(\frac{1-z}{1+z}\right)^{2^m} \right] \nonumber \\
&=&  (1+z)^{2^m} + (1-z)^{2^m} + v (1-z^2)^{2^{m-1}} + \nonumber \\ 
 & &             u(1-z)^{2^{m-1}-2^{(m-1)/2}}(1+z)^{2^{m-1}+2^{(m-1)/2}} + \nonumber \\ 
             & &  
             u(1-z)^{2^{m-1}+2^{(m-1)/2}}(1+z)^{2^{m-1}-2^{(m-1)/2}} .
\end{eqnarray} 

We now treat the terms in (\ref{eqn-j18-1}) one by one. We first have 
\begin{eqnarray}\label{eqn-j18-2}
(1+z)^{2^m} + (1-z)^{2^m} = \sum_{k=0}^{2^m} \left(1+(-1)^k \right) \binom{2^m}{k}.  
\end{eqnarray} 
One can easily see that 
\begin{eqnarray}\label{eqn-j18-3}
(1-z^2)^{2^{m-1}} = \sum_{i=0}^{2^{m-1}} (-1)^i \binom{2^{m-1}}{i} z^{2i} = 
\sum_{k=0}^{2^{m}} \frac{1+(-1)^k}{2} (-1)^{\lfloor k/2 \rfloor} \binom{2^{m-1}}{\lfloor k/2 \rfloor} z^{k}.  
\end{eqnarray} 

Notice that 
\begin{eqnarray*}
(1-z)^{2^{m-1}-2^{(m-1)/2}}=\sum_{i=0}^{2^{m-1}-2^{(m-1)/2}} \binom{2^{m-1}-2^{(m-1)/2}}{i} (-1)^i z^i   
\end{eqnarray*} 
and 
\begin{eqnarray*}
(1+z)^{2^{m-1}+2^{(m-1)/2}}=\sum_{i=0}^{2^{m-1}+2^{(m-1)/2}} \binom{2^{m-1}+2^{(m-1)/2}}{i} z^i   
\end{eqnarray*} 
We have then 
\begin{eqnarray}\label{eqn-j18-4}
\lefteqn{(1-z)^{2^{m-1}-2^{(m-1)/2}} (1+z)^{2^{m-1}+2^{(m-1)/2}} } \nonumber \\
& & = \sum_{k=0}^{2^m} 
\left[ \sum_{\substack{0 \le i \le 2^{m-1}-2^{(m-1)/2} \\
0\le j \le 2^{m-1}+2^{(m-1)/2} \\i+j=k}}(-1)^i \binom{2^{m-1}-2^{(m-1)/2}} {i} \binom{2^{m-1}+2^{(m-1)/2}}{j}     \right] z^k.  
\end{eqnarray}
Similarly, we have 
\begin{eqnarray}\label{eqn-j18-5}
\lefteqn{(1-z)^{2^{m-1}+2^{(m-1)/2}} (1+z)^{2^{m-1}-2^{(m-1)/2}} } \nonumber \\
& & = \sum_{k=0}^{2^m} 
\left[ \sum_{\substack{0 \le i \le 2^{m-1}+2^{(m-1)/2} \\
0\le j \le 2^{m-1}-2^{(m-1)/2} \\i+j=k}}(-1)^i \binom{2^{m-1}+2^{(m-1)/2}} {i} \binom{2^{m-1}-2^{(m-1)/2}}{j}     \right] z^k.  
\end{eqnarray}

Plugging (\ref{eqn-j18-2}), (\ref{eqn-j18-3}), (\ref{eqn-j18-4}), and (\ref{eqn-j18-5}) into 
(\ref{eqn-j18-1}) proves the desired conclusion.  
\end{proof}

\begin{theorem}\label{thm-newdesigns2}
Let $m \geq 5$ be odd. 
Let $\C_m$ be a linear code of length $2^m-1$ such that its dual code
$\C_m^\perp$ has the weight distribution of Table \ref{tab-CG1}. Denote by $\overline{\C}_m$ 
the extended code of $\C_m$ and let $\overline{\C}_m^\perp$ denote the dual of $\overline{\C}_m$.
Let $\cP=\{0,1,2, \cdots, 2^m-1\}$, and let $\overline{\cB}$ be  the set of  the supports of the codewords of $\overline{\C}_m$
with weight $k$, where $\overline{A}_k \neq 0$. Then $(\cP, \overline{\cB})$
is a $3$-$(2^m, k, \lambda)$ design, where
\begin{eqnarray*}
\lambda=\frac{\overline{A}_k\binom{k}{3}}{\binom{2^m}{3}}, 
\end{eqnarray*}
where $\overline{A}_k$ is given in Lemma \ref{lem-TZEwt}. 

Let $\cP=\{0,1,2, \cdots, 2^m-1\}$, and let $\overline{\cB}^\perp$ be  the set of  the supports of the codewords of
$\overline{\C}_m^\perp$
with weight $k$ and $\overline{A}_k^\perp \neq 0$. Then $(\cP, \overline{\cB}^\perp)$
is a $3$-$(2^m, k, \lambda)$ design, where
\begin{eqnarray*}
\lambda=\frac{\overline{A}_k^\perp\binom{k}{3}}{\binom{2^m}{3}}, 
\end{eqnarray*}
where $\overline{A}_k^\perp$ is given in Lemma \ref{lem-TZEwt}. 
\end{theorem}

\begin{proof}
The weight distributions of $\overline{\C}_m$ and $\overline{\C}_m^\perp$ are described in 
Lemma \ref{lem-TZEwt}. 
Notice that the minimum distance $d$ of $\overline{\C}_m$ is equal to $6$. Put $t=3$. The number
of $i$ with $\overline{A}_i^\perp \neq 0$ and $1 \leq i \leq 2^m -t$ is $s=3$. Hence, $s=d-t$.
The desired conclusions then follow from Theorem \ref{thm-AM1} and the fact that two binary 
vectors have the same support if and only if they are identical. 
\end{proof} 

\begin{example} 
Let $m \geq 5$ be odd. 
Then $\overline{\C}_m^\perp$ gives three $3$-designs with the following parameters: 
\begin{itemize}
\item $(v,\, k, \, \lambda)=\left(2^m,\  2^{m-1}-2^{(m-1)/2}, \ 
  (2^{m-3}-2^{(m-3)/2}) (2^{m-1}-2^{(m-1)/2}-1)   \right).$
\item $(v, \, k, \, \lambda)=\left(2^m, \ 2^{m-1}+2^{(m-1)/2}, \ 
  (2^{m-3}+2^{(m-3)/2}) (2^{m-1}-2^{(m-1)/2}-1)   \right).$  
\item $(v, \, k, \, \lambda)=\left(2^m, \ 2^{m-1}, \ 
  (2^{m-1}+1)(2^{m-2}-1) \right).$  
\end{itemize} 
\end{example}

\begin{example} 
Let $m \geq 5$ be odd. 
Then the supports of all codewords of weight $6$ in $\overline{\C}_m$ give a $3$-$(2^m,\, 6,\, \lambda)$ design, where 
$$ 
\lambda= \frac{2^{m-1}-4}{3}. 
$$
\end{example} 

\begin{proof}
By Lemma  \ref{lem-TZEwt}, 
$$ 
\overline{A}_6 = \frac{2^{m-1} (2^{m-1}-1) (2^{m-1}-4)   (2^m-1)}{90}
$$
The desired value for $\lambda$ then follows from Theorem \ref{thm-newdesigns2}. 
\end{proof}

\begin{example} 
Let $m \geq 5$ be odd. 
Then the supports of all codewords of weight $8$ in $\overline{\C}_m$ give a $3$-$(2^m,\, 8,\, \lambda)$ design, where 
$$ 
\lambda= \frac{2 \times 2^{3(m-1)} - 25 \times 2^{2(m-1)} + 123 \times 2^{m-1} - 190}{30}. 
$$
\end{example} 

\begin{proof}
By Lemma  \ref{lem-TZEwt}, 
$$ 
\overline{A}_8 =   \frac{2^{m-1}(2^{m-1}-1)  (2^m-1) (2 \times 2^{3(m-1)} - 25 \times 2^{2(m-1)} + 123 \times 2^{m-1} - 190)}{8 \times 315}. 
$$
The desired value for $\lambda$ then follows from Theorem \ref{thm-newdesigns2}. 
\end{proof}

\begin{example} 
Let $m \geq 5$ be odd. 
Then the supports of all codewords of weight $10$ in $\overline{\C}_m$ give a $3$-$(2^m,\, 10,\, \lambda)$ design, where 
$$ 
\lambda= \frac{ (2^{m-1}-4)  (2 \times 2^{4(m-1)} - 34 \times 2^{3(m-1)} + 235 \times  2^{2(m-1)} - 931 \times  2^{m-1} + 1358)}{315}. 
$$
\end{example} 

\begin{proof}
By Lemma  \ref{lem-TZEwt}, 
$$ 
\overline{A}_{10} = \frac{2^{h} (2^{h}-1) (2^{h}-4)   (2^{h+1}-1) (2\times 2^{4h} - 34 \times  2^{3h} + 235 \times  2^{2h} - 931 \times  2^{h} + 1358)}{4\times 14175}, 
$$
where $h=m-1$. 
The desired value for $\lambda$ then follows from Theorem \ref{thm-newdesigns2}. 
\end{proof}

To demonstrate the existence of the $2$-designs and $3$-designs presented in Theorems \ref{thm-newdesigns1} and \ref{thm-newdesigns2}, respectively, we describe a list of binary codes that 
have the weight distribution of Table \ref{tab-CG1} below.

Let $\alpha$ be a generator of $\gf(2^m)^*$. Let $g_s(x)=\m_1(x)\m_s(x)$, where $\m_i(x)$ is the
minimal polynomial of $\alpha^i$ over $\gf(2)$. Let $\C_m$ denote the cyclic code of length
$v=2^m-1$ over $\gf(2)$ with generator polynomial $g_s(x)$. It is known that $\C_m^\perp$ has 
dimension $2m$ and the weight distribution of Table \ref{tab-CG1} when $m$ is odd and $s$ takes on the
following values \cite{DLLZ}:
\begin{enumerate}
\item $s=2^h+1$, where $\gcd(h, m)=1$ and $h$ is a positive integer.
\item $s=2^{2h}-2^h+1$, where $h$ is a positive integer.
\item $s=2^{(m-1)/2}+3$.
\item $s=2^{(m-1)/2}+2^{(m-1)/4}-1$, where $m \equiv 1 \pmod{4}$.
\item $s=2^{(m-1)/2}+2^{(3m-1)/4}-1$, where $m \equiv 3 \pmod{4}$.
\end{enumerate}
In all these cases, $\C_m$ has parameters $[2^m-1, 2^m-1-2m, 5]$ and is optimal.
It is also known that the binary narrow-sense primitive BCH code with designed distance
$2^{m-1}-2^{(m-1)/2}$ has also the weight distribution of Table \ref{tab-CG1} \cite{DFZ}.
These codes
and their extended codes give $2$-designs and $3$-designs when they are plugged into Theorems
\ref{thm-newdesigns1} and \ref{thm-newdesigns2}.

It is known that $\C_m$ has parameters $[2^m-1, 2^m-1-2m, 5]$ if and only if $x^e$ is an APN
monomial over $\gf(2^m)$. However, even if $x^e$ is APN, the dual code $\C_m^\perp$ may have
many weights, and thus the code $\C_m$ and its dual $\C_m^\perp$ may not give $2$-designs.
One of such examples is the inverse APN monomial.

%\begin{open} 
%Are the $2$-designs (respectively, $3$-designs) with the same parameters defined by these binary 
%linear codes with generator polynomials $g_s(x)$ equivalent? 
%\end{open} 

\section{Infinite families of $2$-designs from a type of ternary linear codes}\label{sec-june28}

In this section, we will construct infinite families of $2$-designs with a type of primitive 
ternary cyclic codes.  

\begin{table}[ht]
\caption{Weight distribution of some ternary linear codes}\label{tab-CG3}
\centering
\begin{tabular}{|l|l|}
\hline
Weight $w$    & No. of codewords $A_w$  \\ \hline
$0$                                                        & $1$  \\ 
$2\times 3^{m-1}-3^{(m-1)/2}$           & $(3^m-1)(3^{m-1}+3^{(m-1)/2})$  \\ 
$2\times 3^{m-1}$                            & $(3^m-1)(3^{m-1}+1)$ \\ 
$2\times 3^{m-1}+3^{(m-1)/2}$          & $(3^m-1)(3^{m-1}-3^{(m-1)/2})$ \\ \hline
\end{tabular}
\end{table}

\begin{table}[ht]
\caption{Weight distribution of some ternary linear codes}\label{tab-CG328}
\centering
\begin{tabular}{|l|l|}
\hline
Weight $w$    & No. of codewords $A_w$  \\ \hline
$0$                                                        & $1$  \\ 
$2\times 3^{m-1}-3^{(m-1)/2}$           & $3^{2m}-3^m$  \\ 
$2\times 3^{m-1}$                            & $(3^m+3)(3^m-1)$ \\ 
$2\times 3^{m-1}+3^{(m-1)/2}$          & $3^{2m}-3^m$ \\ \hline
$3^m$                                  &  $2$ \\ \hline 
\end{tabular}
\end{table}

\begin{lemma}\label{lem-TZEwt28} 
Let $m \geq 3$ be odd. Assume that $\C_m$ is a ternary linear code of length $3^m-1$ such that 
its dual code $\C_m^\perp$ has the weight distribution of Table \ref{tab-CG3}. Denote by 
$\overline{\C}_m$ the extended code of $\C_m$ and let $\overline{\C}_m^\perp$ denote the dual 
of $\overline{\C}_m$. 
Then we have the following conclusions. 
\begin{enumerate}
\item The code $\C_m$ has parameters $[3^m-1, \, 3^m-1-2m, \, 4]$.  
\item The code $\C_m^\perp$ has parameters $[3^m-1, \, 2m, \, 2\times 3^{m-1}-3^{(m-1)/2}]$.  
\item The code $\overline{\C}_m^\perp$ has parameters $[3^m, \, 2m+1, \, 2\times 3^{m-1}-3^{(m-1)/2}]$, 
and its 
      weight distribution is given in Table \ref{tab-CG328}.    
\item The code $\overline{\C}_m$ has parameters $[3^m, 3^m-1-2m, 5]$, and its weight distribution
is given by 
\begin{eqnarray*}
3^{2m+1}\overline{A}_k &=&   (2^k+(-1)^k 2) \binom{3^m}{k} + \\ 
& & v \sum_{\substack{0 \le i \le 2 \times 3^{m-1} \\
0\le j \le 3^{m-1} \\i+j=k}}(-1)^i \binom{2 \times 3^{m-1}} {i} 2^j \binom{3^{m-1}}{j}   + \\ 
& & u \sum_{\substack{0 \le i \le 2\times 3^{m-1}-3^{\frac{m-1}{2}} \\
0\le j \le 3^{m-1}+3^{\frac{m-1}{2}} \\i+j=k}}(-1)^i \binom{2 \times 3^{m-1}-3^{\frac{m-1}{2}}} {i} 2^j \binom{3^{m-1}+3^{\frac{m-1}{2}}}{j} + \\
& & u  
\sum_{\substack{0 \le i \le 2 \times 3^{m-1}+3^{\frac{m-1}{2} } \\
0\le j \le 3^{m-1}-3^{\frac{m-1}{2}} \\i+j=k}}(-1)^i \binom{2 \times 3^{m-1}+3^{\frac{m-1}{2}}}{i}2^j \binom{3^{m-1}-3^{\frac{m-1}{2}}}{j} 
\end{eqnarray*}
for $0 \le k \le 3^m$, where 
$$ 
u=3^{2m}-3^{m} \mbox{ and } v = (3^m+3)(3^m-1). 
$$
\end{enumerate}
\end{lemma}

\begin{proof}
The proof is similar to that of Lemma \ref{lem-TZEwt} and is omitted here. 
\end{proof}

\begin{theorem}\label{thm-newdesigns228}
Let $m \geq 3$ be odd. 
Let $\C_m$ be a linear code of length $3^m-1$ such that its dual code
$\C_m^\perp$ has the weight distribution of Table \ref{tab-CG3}. Denote by $\overline{\C}_m$ the 
extended code of $\C_m$ and let $\overline{\C}_m^\perp$ denote the dual of $\overline{\C}_m$.
Let $\cP=\{0,1,2, \cdots, 3^m-1\}$, and let $\overline{\cB}$ be  the set of the supports of the codewords of $\overline{\C}_m$
with weight $k$, where $5 \leq k \leq 10$ and $\overline{A}_k \neq 0$. Then $(\cP, \overline{\cB})$
is a $2$-$(3^m,\, k,\, \lambda)$ design for some $\lambda$. 

Let $\cP=\{0,1,2, \cdots, 3^m-1\}$, and let $\overline{\cB}^\perp$ be  the set of the supports of the codewords of
$\overline{\C}_m^\perp$
with weight $k$ and $\overline{A}_k^\perp \neq 0$. Then $(\cP, \overline{\cB}^\perp)$
is a $2$-$(3^m, \, k, \, \lambda)$ design for some $\lambda$. 
\end{theorem}

\begin{proof}
The weight distributions of $\overline{\C}_m$ and $\overline{\C}_m^\perp$ are described in 
Lemma \ref{lem-TZEwt28}. 
Notice that the minimum distance $d$ of $\overline{\C}_m$ is equal to $5$. Put $t=2$. The number
of $i$ with $\overline{A}_i^\perp \neq 0$ and $1 \leq i \leq 3^m -t$ is $s=3$. Hence, $s=d-t$.
The desired conclusions then follow from Theorem \ref{thm-AM2}. 
\end{proof}

\begin{corollary} 
Let $m \geq 3$ be odd. 
Let $\C_m$ be a ternary linear code of length $3^m-1$ such that its dual code
$\C_m^\perp$ has the weight distribution of Table \ref{tab-CG3}. Denote by $\overline{\C}_m$ 
the extended code of $\C_m$ and let $\overline{\C}_m^\perp$ denote the dual of $\overline{\C}_m$.

Let $\cP=\{0,1,2, \cdots, 3^m-1\}$, and let $\overline{\cB}^\perp$ be  the set of the supports of the codewords of
$\overline{\C}_m^\perp$
with weight $2\times 3^{m-1}-3^{(m-1)/2}$. Then $(\cP, \overline{\cB}^\perp)$
is a $2$-$(3^m, \, 2\times 3^{m-1}-3^{(m-1)/2}, \, \lambda)$, where 
$$ 
\lambda=\frac{(2\times 3^{m-1}- 3^{(m-1)/2})(2\times 3^{m-1}- 3^{(m-1)/2} -1)}{2}. 
$$ 
\end{corollary}

\begin{proof}
It follows from Theorem \ref{thm-newdesigns228} that $(\cP, \overline{\cB}^\perp)$ is a $2$-design. 
We now determine the value of $\lambda$. Note that $\overline{\C}_m^\perp$ has minimum weight 
$2\times 3^{m-1}-3^{(m-1)/2}$. Any two codewords of minimum weight $2\times 3^{m-1}-3^{(m-1)/2}$ 
have the same support if and only if one is a scalar multiple of the other. Consequently, 
$$ 
\left|\overline{\cB}^\perp \right|=\frac{3^{2m}-3^m}{2}. 
$$
It then follows that 
$$ 
\lambda=\frac{3^{2m}-3^m}{2}\frac{\binom{2\times 3^{m-1}-3^{(m-1)/2}}{2}}{\binom{3^m}{2}}
= \frac{(2\times 3^{m-1}- 3^{(m-1)/2})(2\times 3^{m-1}- 3^{(m-1)/2} -1)}{2}.
$$
\end{proof}

\begin{corollary} 
Let $m \geq 3$ be odd. 
Let $\C_m$ be a ternary linear code of length $3^m-1$ such that its dual code
$\C_m^\perp$ has the weight distribution of Table \ref{tab-CG3}. Denote by $\overline{\C}_m$ 
the extended code of $\C_m$ and let $\overline{\C}_m^\perp$ denote the dual of $\overline{\C}_m$.
Let $\cP=\{0,1,2, \cdots, 3^m-1\}$, and let $\overline{\cB}$ be  the set of the supports of the codewords of $\overline{\C}_m$
with weight $5$. Then $(\cP, \overline{\cB})$
is a $2$-$(3^m,\, 5,\, \lambda)$ design, where  
$$ 
\lambda=\frac{5(3^{m-1}-1)}{2}.  
$$ 
\end{corollary} 

\begin{proof}
It follows from Theorem \ref{thm-newdesigns228} that $(\cP, \overline{\cB})$ is a $2$-design. 
We now determine the value of $\lambda$. Using the weight distribution formula in Lemma 
\ref{lem-TZEwt28}, we obtain that 
$$ 
\overline{A}_5=\frac{3^{3m-1}-4 \times 3^{2m-1}+3^m}{4}. 
$$ 
Recall that $\overline{\C}_m$ has minimum weight 
$5$. Any two codewords of minimum weight $5$ have the same support if and only if one is a 
scalar multiple of the other. Consequently, 
$$ 
\left|\overline{\cB}^\perp \right|=\frac{\overline{A}_5}{2}. 
$$
It then follows that 
$$ 
\lambda= \frac{\overline{A}_5}{2} \frac{\binom{5}{2}}{\binom{3^m}{2}}=\frac{5(3^{m-1}-1)}{2}.  
$$
\end{proof} 

Theorem \ref{thm-newdesigns228} gives more $2$-designs. However, determining the corresponding 
value $\lambda$ may be hard, as the number of blocks in the design may be difficult to derive  
from $\overline{A}_k$ or $\overline{A}_k^\perp$. 

\begin{open} 
Determine the value of $\lambda$ of the $2$-$(3^m,\, k,\, \lambda)$ design for  $6 \leq k \leq 10$,  
which are described in Theorem \ref{thm-newdesigns228}. 
\end{open} 

\begin{open} 
Determine the values of $\lambda$ of the $2$-$(3^m,\, 3^{m-1},\, \lambda)$ design and the   
$2$-$(3^m,\, 2\times 3^{m-1}-3^{(m-1)/2},\, \lambda)$ design, 
which are described in Theorem \ref{thm-newdesigns228}. 
\end{open}

To demonstrate the existence of the $2$-designs presented in Theorem \ref{thm-newdesigns228}, 
we present a list of ternary cyclic codes that have the weight distribution of Table \ref{tab-CG3} below.

Put $n=3^m-1$. 
Let $\alpha$ be a generator of $\gf(3^m)^*$. Let $g_s(x)=\m_{n-1}(x)\m_{n-s}(x)$, where $\m_i(x)$ 
is the minimal polynomial of $\alpha^i$ over $\gf(3)$. Let $\C_m$ denote the cyclic code of length
$n=3^m-1$ over $\gf(3)$ with generator polynomial $g_s(x)$. It is known that $\C_m^\perp$ has 
dimension $2m$ and the weight distribution of Table \ref{tab-CG3} when $m$ is odd and $s$ takes on the
following values \cite{CDY,YCD}:
\begin{enumerate}
\item $s=3^h+1$, $h \geq 0$ is an integer.
\item $s=(3^h + 1)/2$, where $h$ is a positive integer and $\gcd(m, h)=1$.
\end{enumerate}
In these two cases, $x^s$ is a planar function on $\gf(3^m)$. Hence, these ternary codes are 
extremal in the sense that they are defined by planar functions whose differentiality is 
extremal. 

More classes of ternary codes such that their duals have the weight distribution 
of  Table \ref{tab-CG3} are documented in \cite{DLLZ}. They give also $2$-designs 
via Theorem \ref{thm-newdesigns228}. There are also ternary cyclic codes with three weights 
but different weight distributions in \cite{DLLZ}. They may also hold $2$-designs.

\section{Conjectured infinite families of $2$-designs from projective cyclic codes}\label{sec-conjectureddesigns}

Throughout this section, let $m \geq 3$ be an odd integer, and let $v=(3^m-1)/2$. The objective of 
this section is to present a number of conjectured infinite families of $2$-designs derived from 
linear projective ternary cyclic codes.

\begin{table}
\begin{center}
\caption{The weight distribution for odd  $m \ge3$}\label{Tab-GG2}
\begin{tabular}{|c|c|}
\hline
Weight  &  Frequency \\ \hline
$0$     &  $1$ \\ \hline
$3^{m-1}-3^{(m-1)/2}$ & $\frac{(3^{m-1}+3^{(m-1)/2})(3^m-1)}{2}$ \\ \hline
$3^{m-1}$ & $(3^m-3^{m-1}+1)(3^m-1)$ \\ \hline
$3^{m-1}+3^{(m-1)/2}$ & $\frac{(3^{m-1}-3^{(m-1)/2})(3^m-1)}{2}$ \\ \hline
\end{tabular}
\label{table4}
\end{center}
\end{table}

\begin{lemma}\label{lem-june21}
Let $\C_m$ be a linear code of length $v$ over $\gf(3)$ such that its dual $\C_m^\perp$ has the 
weight distribution in Table \ref{Tab-GG2}. Then the weight distribution of $\C_m$ is given by
\begin{eqnarray*}
3^{2m}A_k &=&  \sum_{\substack{0 \le i \le 3^{m-1}-3^{(m-1)/2} \\
0\le j \le \frac {3^{m-1}+2\cdot 3^{(m-1)/2}-1} {2} \\i+j=k}}(-1)^i2^ja\binom{3^{m-1}-3^{(m-1)/2}} {i} \binom{\frac {3^{m-1}+2 \cdot 3^{(m-1)/2}-1}{2}}{j}\\
& & +  \binom {\frac {3^m-1} {2}}{k}2^k+ \sum_{\substack{0 \le i \le 3^{m-1} \\
0\le j \le \frac {3^{m-1}-1} {2} \\i+j=k}}(-1)^i2^jb\binom{3^{m-1}} {i}\binom{\frac {3^{m-1}-1} {2}} {j} \\ & & + 
\sum_{\substack{0 \le i \le 3^{m-1}+3^{(m-1)/2} \\
0\le j \le \frac {3^{m-1}-2 \cdot 3^{(m-1)/2}-1}{2} \\i+j=k}}(-1)^i2^jc\binom{3^{m-1}+3^{(m-1)/2}}{i}\binom{\frac {3^{m-1}-2 \cdot 3^{(m-1)/2}-1}{2}}{j} 
\end{eqnarray*}
for $0 \le k \le \frac {3^m-1} {2}$, where 
\begin{eqnarray*}
a &=& \frac{(3^{m-1}+3^{(m-1)/2})(3^m-1)}{2}, \\ 
b &=& (3^m-3^{m-1}+1)(3^m-1), \\
c &=& \frac{(3^{m-1}-3^{(m-1)/2})(3^m-1)}{2}.
\end{eqnarray*} 
In addition, $\C_m$ has parameters $[(3^m-1)/2, (3^m-1)/2-2m, 4]$. 
\end{lemma}

\begin{proof}
Note that the weight enumerator of $\C_m^\perp$ is 
$$1+az^{3^{m-1}-3^{(m-1)/2}}+bz^{3^{m-1}}+cz^{3^{m-1}+3^{(m-1)/2}}.$$
The proof of this theorem is similar to that of Lemma \ref{lem-TZwt} and is omitted.
\end{proof}

Below we present two examples of ternary linear codes $\C_m$ such that their duals $\C_m^\perp$ 
have the weight distribution of Table \ref{Tab-GG2}. 

\begin{example}\label{exam-j271} 
Let $m \geq 3$ be odd. Let $\alpha$ be a generator of $\gf(3^m)^*$. Put $\beta=\alpha^2$. Let $\m_i(x)$ denote 
the minimal polynomial of $\beta^i$ over $\gf(3)$. Define 
$$ 
\delta=3^{m-1}-1-\frac{3^{(m+1)/2}-1}{2} 
$$
and 
$$ 
h(x)=(x-1)\lcm(\m_1(x), \, \m_2(x), \, \cdots, \, \m_{\delta-1}(x)), 
$$
where $\lcm$ denotes the least common multiple of the polynomials. 
Let $\C_m$ denote the cyclic code of length $v=(3^m-1)/2$ over $\gf(3)$ with 
generator polynomial $g(x):=(x^v-1)/h(x)$. Then $\C_m$ has parameters $[(3^m-1)/2, (3^m-1)-2m, 4]$ and 
$\C_m^\perp$ has the weight distribution of Table \ref{Tab-GG2}.  
\end{example}    

\begin{proof}
A proof of the desired conclusions was given in \cite{LDXG}. 
\end{proof}

\begin{example}\label{exam-j272}  
Let $m \geq 3$ be odd. Let $\alpha$ be a generator of $\gf(3^m)^*$. Let $\beta=\alpha^2$.
Let $g(x)=\m_{n-1}(x)\m_{n-2}(x)$, where $\m_i(x)$ is the minimal polynomial of $\beta^i$ over
$\gf(3)$. Let $\C_m$ denote the cyclic code of length $v=(3^m-1)/2$ over $\gf(3)$ with 
generator polynomial $g(x)$. Then $\C_m$ has parameters $[(3^m-1)/2, (3^m-1)-2m, 4]$ and 
$\C_m^\perp$ has the weight distribution of Table \ref{Tab-GG2}.  
\end{example}    

\begin{proof}
The desired conclusions can be proved similarly as Theorem 19 in \cite{LDXG}.  
\end{proof}

\begin{conj}\label{conj-j261}
Let $\cP=\{0,1,2, \cdots, v-1\}$, and let $\cB$ be  the set of the supports of the codewords of $\C_m$
with Hamming weight $k$, where $A_k \neq 0$. Then $(\cP, \cB)$
is a $2$-$(v, k, \lambda)$ design for all odd $m \ge 3$. 
\end{conj}

\begin{conj}\label{conj-j262}
Let $\cP=\{0,1,2, \cdots, v-1\}$, and let $\cB$ be  the set of the supports of the codewords of $\C_m$
with Hamming weight $4$. Then $(\cP, \cB)$
is a Steiner system $S(2, 4, (3^m-1)/2)$ for all odd $m \ge 3$. 
\end{conj}

There are a survey on Steiner systems $S(2,4,v)$ \cite{RR10} and a book chapter on Steiner 
systems \cite{CMhb}. It is known that a Steiner system $S(2,4,v)$ exists if and only if 
$v \equiv 1 \mbox{ or } 4 \pmod{12}$ \cite{Hanani}.  

If Conjecture \ref{conj-j261} is true, so is Conjecture \ref{conj-j262}. In this case, a 
coding theory construction of a Steiner system $S(2, 4, (3^m-1)/2)$ for all odd $m \ge 3$ 
is obtained.

\begin{conj}\label{conj-j263}
Let $\cP=\{0,1,2, \cdots, v-1\}$, and let $\cB$ be  the set of the supports of the codewords of $\C_m^\perp$
with Hamming weight $k$, where $A_k^\perp \neq 0$. Then $(\cP, \cB)$
is a $2$-$(v, k, \lambda)$ design for all odd $m \ge 3$. 
\end{conj}

Even if some or all of the three conjectures are not true for ternary codes with the weight 
distribution of Table \ref{Tab-GG2}, these conjectures might still be valid for the two 
classes of ternary cyclic codes descried in Examples \ref{exam-j271} and \ref{exam-j272}.  
Note that Theorem \ref{thm-AM2} does not apply to the three conjectures above. We need to 
develop different methods for settling these conjectures. 

\section{Summary and concluding remarks} 

In the last section of this paper, we mention some applications of $t$-designs and summarize the main contributions 
of this paper.  

\subsection{Some applications of $2$-designs} 

Let $\cP$ be an Abelian group of order $v$ under a binary operation denoted by $+$. Let 
$\cB=\{B_1, B_2, \cdots, B_b\}$, where all $B_i$ are $k$-subsets of $\cP$ and $k$ is a 
positive integer. We define $\Delta(B_i)$ to be the multiset $\{ x-y: x \in B_i,\  y \in B_i\}$.  
If every nonzero element of $\cP$ appears exactly $\delta$ times in the multiset 
$\bigcup_{i=1}^b \Delta(B_i)$, we call $\cB$ a $(v, k, \delta)$ difference family in 
$(\cP, +)$. 

The following theorems are straightforward and should be well known. 

\begin{theorem}\label{thm-june261}
Let $\cP$ be an Abelian group of order $v$ under a binary operation denoted by $+$. Let 
$\cB=\{B_1, B_2, \cdots, B_b\}$, where all $B_i$ are $k$-subsets of $\cP$ and $k$ is a 
positive integer. Then $(\cP, \cB)$ is a $2$-$(v, k, \lambda)$ design if and only if 
$\cB$ is a $(v, k, \lambda v)$ difference family in $(\cP, +)$.
\end{theorem}   

\begin{theorem}\label{thm-june262}
Let $(\cP, \cB)$ be a $t$-$(v, k, \lambda)$ design, where $\cP$ is an Abelian group. 
If $t \geq 2$, then $\cB$ is a $(v, k, \delta)$ difference family in $\cP$, 
where 
$$ 
\delta=\frac{v \lambda \binom{v-2}{t-2}}{\binom{k-2}{t-2}}. 
$$ 
\end{theorem}

Difference families have applications in the design and analysis of optical orthogonal codes, 
frequency hopping sequences, and other engineering areas. By Theorems \ref{thm-june261} and 
\ref{thm-june262}, $t$-designs with $t \geq 2$ have also applications in these areas. In addition, 
$2$-designs give naturally linear codes \cite{AK92,DingBook}. These show the importance of 
$2$-designs in applications.     

\subsection{Summary} 

It is well known that binary Reed-Muller codes hold $3$-designs. Hence, the only contribution 
of Section \ref{sec-brmdesigns} is the determination of the specific parameters of the $3$-designs 
held in $\RM(m-2, m)$ and its dual code, which are documented in Theorem \ref{thm-brmdesign1}. 

It has also been known for a long time that the codewords of weight $3$ in the Hamming code 
hold a $2$-$((q^m-1)/(q-1), 3, q-1)$ design. The contribution of Section \ref{sec-hmdesigns} 
is Theorem \ref{thm-HMdesign171}, which may be viewed as an extension of the known 
$2$-$((q^m-1)/(q-1), 3, q-1)$ design held in the Hamming code, and also the parameters 
of the infinite families of $2$-designs derived from the binary Hamming codes, which are 
documented in Examples \ref{exam-hmdesign4}, \ref{exam-hmdesign5}, \ref{exam-hmdesign6}, 
and \ref{exam-hmdesign7}. 

A major contribution of this paper is presented in Section \ref{sec-newdesigns}, where 
Theorems \ref{thm-newdesigns1} and \ref{thm-newdesigns2} document many infinite families of 
$2$-design and $3$-designs. The parameters of these $2$-designs and $3$-designs are given 
specifically. These designs are derived from binary cyclic codes that are defined by special 
almost perfect nonlinear functions.   

Another major contribution of this paper is documented in Section \ref{sec-june28}, 
where Theorem \ref{thm-newdesigns228} and its two corollaries describe several infinite 
families of $2$-designs. These $2$-designs are related to planar functions.    

It is noticed that the total number of $3$-designs presented in this paper (see Theorems  
\ref{thm-brmdesign1} and \ref{thm-newdesigns2}) are exponential. All of them are derived from linear 
codes. After comparing the list of infinite families of $3$-designs in \cite{KLhb} with the $3$-designs 
presented in this paper, one may conclude that many, if not most, of the known infinite families of $3$-designs 
are from coding theory.    

Section \ref{sec-conjectureddesigns} presents many conjectured infinite families of $2$-designs. 
The reader is cordially invited to attack these conjectures and solve other open problems presented 
in this paper.

\end{document}